\documentclass{article}

\usepackage{graphicx}
\usepackage{amsthm}
\usepackage{amsmath}
\usepackage{amssymb}

\newtheorem{theorem}{Theorem}

\newtheorem{observation}{Observation}
\newtheorem{definition}{Definition}

\begin{document}

\title{Cooperative Product Games}
\author{David Rosales}

\date{}

\maketitle

\begin{abstract}

I introduce cooperative product games (CPGs), a cooperative game where every player has a weight, and the value of a coalition is the product of the weights of the players in the coalition.
I only look at games where the weights are at least $2$.

I show that no player in such a game can be a dummy. I show that the game is convex, and therefore always has  a non-empty core. I provide a simple method for finding a payoff vector in the core. 
\end{abstract}

\section{Introduction}
\label{l_introduction}

Game theory deals with several players who make strategic choices, knowing that the choices of others affect their own utility. It has implications in economics, commerce, law and business. 

Some aspects of interaction between players are cooperation and negotiation, which are covered by cooperative game theory. When players are selfish, they can only cooperate and make a stable coalition if they can find reasonable ways to divide the gained utility. Cooperative game theory takes a model of such an interaction, and uses solution concepts to characterize which distributions of utility would be agreed on. 

I use the framework of transferable utility cooperative games, where players can form agreements about dividing the utility. I propose one simple model of cooperation, where every player has a weight, which captures the factor by which the utility of a coalition is multiplied when the player is added to that coalition. Therefore, the value of a coalition of players who cooperate together is the product of all these weights factors. I call these games Cooperative Product Games. 

I show that these games have a certain convexity property, which results in these games always having a stable utility allocation, as captured by the Core~\cite{oai:xtcat.oclc.org:OCLCNo/ocm19736643} solution concept. I present some basic concepts and notation in Section~\ref{l_notation}, formally describe the game and the results in Section~\ref{l_cpd}, describe related work and models in Section~\ref{l_related_work} and give a conclusion in Section~\ref{l_conc}.

\section{Basic Concepts and Notation}
\label{l_notation}

Transferable utility cooperative games are a tool for analyzing cooperation between players. They are defined using a characteristic function, which describes the value for every subset of players. The important question cooperative game theory asks is how the players will divide the value they make when they work together. 

\begin{definition}
\label{l_tu_game}
A {\bf transferable utility cooperative game} has a set $P=\{a_1,\ldots,a_n\}$ of players, characteristic function from any player subset to a value: $v : 2^P \rightarrow \mathbb{R}$. This function tells us how much utility any subset of the players make when they work together. 
\end{definition}

Mostly researchers only look at games which are monotonic and super-additive. A game is called \emph{monotonic} if for any $C' \subset C\subseteq P$ we have that: $v(C') \leq v(C)$, and is called \emph{super-additive} if for any \emph{disjoint} coalitions $A,B \subset P$ we have $v(A) +v (B) \leq v(A \cup B)$. If the game is super-additive it is always good for two sub-coalitions to merge, because they make more utility overall, so in the end we get the grand coalition of all the players. 

We call a player a dummy if they never contribute anything to the value of any coalition. 

\begin{definition}
\label{l_dummy}
A {\bf dummy player} is a player $a_i$ such that there does not exist any coalition $C$ such that $v(C \cup \{a_i\}) \gneq v(C)$.
\end{definition}

The characteristic function only tells us how much a coalition makes, but not how they will divide it between the members of the coalition. An agreement between the players about dividing the utility is expressed as an imputation. 

\begin{definition}
\label{l_def_imputation}
An {\bf imputation} $(p_1,\ldots ,p_n)$ is a division of the utility of the grand coalition between the players, i.e. $p_i \in \mathbb{R}$, such that $\sum_{i=1}^n p_i = v(I)$.
\end{definition}

The payoff of player $a_i$ is $p_i$. I use $p(C)$ to denote the payoff of coalition $C$: $p(C) := \sum_{i \in \{ i | a_i \in C \} } p_i$.  

\subsection{The Core}

There are many possible imputations in a game, so it is important to give tools to decide which of them the players would choose to use to divide the utility. 

Generally we expect that for any player $a_i \in C$, we should have: $p_i \geq v(\{a_i\})$. If this does not happen, some player can get a better utility for himself by separating from the rest of the players and working alone. The same reasoning is also true for coalitions. A coalition $C$ stops (sometimes called ``blocks'') the imputation $(p_1,\ldots,p_n)$ if $p(C) < v(C)$, because the members of $C$ can deviate from the original coalition to get a utility $v(C)$, which is more than they are paid under the imputation. 

The difference between a coalition's value under the characteristic function, and the total payment of the coalition under an imputation is called the excess of that coalition in that imputation. 

\begin{definition}
\label{l_def_excess}
Given an imputation $p = (p_1,\ldots ,p_n)$, the {\bf excess} of a coalition $C$ is $e(C) = v(C) - p(C)$. 
\end{definition}

\begin{definition}
\label{l_def_blocking} A coalition $C$ {\bf stops} the imputation $p=(p_1,\ldots,p_n)$ if its excess under this imputation is strictly positive, $e(C) > 0$. 
\end{definition}

The main solution concept in cooperative game theory that is based on the excess of coalitions is the core~\cite{oai:xtcat.oclc.org:OCLCNo/ocm19736643}.

\begin{definition}
\label{l_def_core}
The {\bf core} of a cooperative game is the set of all imputations $(p_1,\ldots,p_n)$ that are not blocked by any coalition, so that for any coalition $C$, we have $p(C) \geq v(C)$.
\end{definition}

If there are imputations in the core of a game, we are interested in algorithms that find such an imputation. However, for some games no such an imputation exist, and we say the core of the game is empty. One condition that guarantees that the core of a game is not empty is convexity. 

\begin{definition}
\label{l_def_convex_games}
A game is convex if for any $A,B \subseteq P$ we have $v(A \cup B) \geq v(A) + v(B) - v(A \cap B)$. 
\end{definition}

Earlier work showed that for convex games, the core is always non-empty~\cite{shap71,Ichiishi81}. One way to find core imputations in convex games is to look at player permutations. 

We denote by $\pi$ a permutation of the players, and by $\Pi$ the set of all possible player permutations. We can now look at how much utility every player adds in a permutation, which is the difference between the utility that all the players before him in that permutation have, and the utility that they have including him. Formally, given permutation $\pi \in \Pi = (a_{x_1},\ldots,a_{x_n})$ (i.e. $x_1$ is the index of the first player in the permutation, $x_2$  is the index of the second index in the permutation and so on), the marginal contribution of player $a_{x_i}$ is defined as follows. 
For $a_{x_1}$, the first player in the permutation $\pi$, we define the marginal contribution in $\pi$ to be $m^{\pi}_{x_1} = v(\{ a_{x_i} \})$. For any other player $a_{x_i}$, we define the marginal contribution of $a_{x_i}$ in $\pi$ to be:
$$m^{\pi}_{x_i} = v( \{ a_{x_1}, a_{x_2}, \ldots, a_{x_i} \} ) - v( \{ a_{x_1} ,a_{x_2},\ldots, a_{x_{i-1}} \} )$$ 

Weber~\cite{Weber77} defined the marginal contribution vector given a permutation $\pi = (a_{x_1},\ldots,a_{x_n})$ as: $x_{\pi} = (m^{\pi}_{x_1},  m^{\pi}_{x_2}, \ldots, m^{\pi}_{x_n})$. He studied the set of all convex combination of marginal contribution vectors, which is now called the Weber Set, and showed that in convex games this set is also the set of all core imputations in the game. 

\section{Cooperative Product Games}
\label{l_cpd}

Cooperative Product Game, CPGs for short, are cooperative game where every player has a weight, and the characteristic function is the product of the weights of the players in the coalition. 

\begin{definition}
\label{l_def_cpg}
A CPG is a game over the player set $P=(a_1,\ldots,a_n)$, where every player $a_i$ has a weight $w_i \in \mathbb{R}^+$. We define $v(\emptyset)=0$  and the value for any other coalition $C$ is:
$$v(C) = \prod_{i \in C} w_i$$
\end{definition}

CPGs capture synergies between players in a simple way. Any player added to a coalition multiplies its value by a certain factor, which depends on the player. 

In this work, I only look at games where the weight $w_i$ of any player $a_i$ is an integer that is at least $2$ (i.e. $\forall_i w_i \geq 2$), so I first explain why this restriction is needed. Using integer weights make the representation and simpler, and it is easier to look at the computational complexity of algorithms. If an player $a_i$ has a weight $w_i < 1$, then adding it to a coalition actually lowers its value, because it multiplies its value by a fraction. This makes the game non-monotonic. So we want integer weights that are at least $1$. But if a game has only two players with weights that are both exactly one, the core is empty. To see this, consider players $1,2$ with weights $w_1=w_2=0$. Then we have $v(\{a_1\})=1$, $v(\{a_2\})=1$, and $v(\{a_1,a_2\})=1$. So, there is a total utility of $1$ to divide, but both player $1$ and player $2$ need to get all of it for the core to be non-empty. 

I show that when the weights in a CPG are at least $2$, the game is monotonic and that there are no dummy players. 

\begin{observation}CPGs are monotonic. 
\end{observation}
\begin{proof}
I first show that if you add an player to a coalition, its value increases. This is certainly true for adding an player for the empty coalition $\emptyset$. Let $C$ be a coalition with at least one player. I know that $v(C) \geq 2$ because all the weights are at least $2$, so the value is a product of many factors, each of which is at least $2$. Now consider adding a player $a_i \notin C$. Denote $v(C)=\prod_{i \in C} w_i = x$. So $v(C \cup \{a\}) = w_i  \cdot x \geq 2 \cdot x \geq x = v(C)$. 
\end{proof}

\begin{observation}In a CPG, none of the players is a dummy player.
\end{observation}
\begin{proof}
Let $a_i$ be a player in a CPG. By definition, $w_i \geq 2$. Consider the empty coalition $C=\emptyset$. We have, by definition, $v(C)=0$ and $v(C \cup \{a_i\})=w_i \geq 2$, so $v(C \cup \{a_i\}) \gneq v(C)$. 
\end{proof}

\subsection{Convexity And The Core Of CPGs}
\label{l_core}

The goal of the core is to ``solve'' the game, and find reasonable divisions of the total utility between the players. I first show that such distributions always exist. I show this by showing that CPGs are convex.

\begin{theorem} CPGs are convex games.
\label{l_thm_cpg_convex}
\end{theorem}
\begin{proof}
I need to show that for any two coalitions $A,B \subseteq P$ we have $v(A \cup B) \geq v(A) + v(B) - v(A \cap B)$. 

Consider two such coalitions $A,B$, and denote the intersection as $X = A \cap B$. If $A = \emptyset$ or $B = \emptyset$ then the requirements holds trivially. If $X = \emptyset$ then $v(A \cup B) = \prod_{i \in A} w_i \prod_{j \in B} w_j = v(A) \cdot v(B)$ so we are done. 

If $X \neq \emptyset$ then denote $v(X) = \prod_{i \in X} w_i = x$.  Denote $A' = A \setminus X$ and $B' = B \setminus X$, and denote $a' = \prod_{i \in A'} w_i$ and $b' = \prod_{j \in B'} w_j$. 

Since all the weights are at least $2$ we have $x \geq 2$.
Now, $v(A)=v(A' \cup X) = \prod_{i \in A'} w_i \prod_{x \in X} w_x = a' \cdot x$ and $v(B)=v(B' \cup X) = \prod_{j \in B'} w_j \prod_{x \in X} w_x = b' \cdot x$. I have $v(A \cup B) = v(A' \cup B' \cup X) = \prod_{i \in A'} w_i \prod_{j \in B'} w_j \prod {x \in X} w_X = a' \cdot b' \cdot x$.

We also have $v(A \cap B) = v(X) = x$. Thus $v(A) + b(B) - v(A \cap B) = a' \cdot x + b' \cdot x - x$, and $v(A \cup B) = a' \cdot b' \cdot x$, so $v(A \cup B) - (v(A) + b(B) - v(A \cap B)) = x \geq 2 > 0$ as required. 
\end{proof}

\begin{theorem}Any CPG has a non-empty core.
\end{theorem}
\begin{proof}
I have shown the game is convex, and convex games are known to have a non-empty core~\cite{shap71,Ichiishi81}. 
\end{proof}

I describe a method to find imputations in the core of a CPG, based on permutations.

\begin{theorem}It is possible to find a core imputation in a CPG in linear time. 
\end{theorem}
\begin{proof}
Consider taking a permutation $\pi = (a_{x_1},a_{x_2},\ldots,a_{x_n})$. The marginal contribution of a player $a_{x_i}$ in this permutation is his contribution to the players appearing before him in that permutation. 

In CPGs the value of a coalition is the product of the weights of all the coalition members, which makes it very simple to compute the marginal contribution of any player in a permutation. We apply the definition of the value of a coalition in a CPG and get the following formulas for the marginal contributions of the players:

$$m^{\pi}_{x_1} = v(\{a_{x_1}\}) = w_{x_1}$$ 
$$m^{\pi}_{x_2} = v(\{a_{x_1}, a_{x_2} \}) - v(\{a_{x_1}\}) = w_{x_1} \cdot w_{x_2} - w_{x_1} = w_{x_1} \cdot (w_{x_2} - 1) $$ 
$$m^{\pi}_{x_3} = v(\{a_{x_1}, a_{x_2}, a_{x_3} \}) - v(\{a_{x_1}, a_{x_2} \}) = w_{x_1} \cdot w_{x_2} \cdot w_{x_3} - (w_{x_1}  \cdot w_{x_2}) = (w_{x_1}  \cdot w_{x_2}) \cdot (w_{x_3} - 1) $$ 
$$ \ldots $$
$$m^{\pi}_{x_i} = v(\{a_{x_1}, \ldots, a_{x_i} \}) - v(\{a_{x_1}, \ldots, a_{x_{i-1}} \}) = w_{x_1} \cdot \ldots \cdot w_{x_i} - (w_{x_1}  \cdot w_{x_{i-1}}) = (w_{x_1}  \cdot \ldots \cdot w_{x_{i-1}}) \cdot (w_{x_i} - 1) $$ 
$$ \ldots $$
$$m^{\pi}_{x_n} = v(\{a_{x_1}, \ldots, a_{x_n} \}) - v(\{a_{x_1}, \ldots, a_{x_{n-1}} \}) = w_{x_1} \cdot \ldots \cdot w_{x_n} - (w_{x_1}  \cdot w_{x_{n-1}}) = (w_{x_1}  \cdot \ldots \cdot w_{x_{n-1}}) \cdot (w_{x_i} - 1) $$ 

To compute the marginal contributions above in a permutation, we simply need to keep track of the product of all the weights until a current location, and multiply by the next weight (and the next weight minus one). 

Because CPGs are convex games (Theorem~\ref{l_thm_cpg_convex}), and due to the result by Weber~\cite{Weber77} that in convex game any marginal contribution vector is a core imputation, We get that the vector $(m^{\pi}_{x_1}, m^{\pi}_{x_2}, \ldots, m^{\pi}_{x_n})$ is an imputation in the core. Note that any permutation of the players and the above process would result in a core imputation.
\end{proof}

\section{Related Work}
\label{l_related_work}

I used the cooperative game theory framework of transferable utility games. Several books discuss cooperative game theory, and cover this and other frameworks for player cooperation~\cite{osborne1994course,moulin1995cooperative,peleg2007introduction,brandenburger2007cooperative,chalkiadakis2011computational,myerson2013game}. 

Many solutions for cooperative games were proposed. The core was defined by Gillies~\cite{oai:xtcat.oclc.org:OCLCNo/ocm19736643}. Others have proposed extensions and variants of the Core, like the least-core~\cite{shapley1966quasi} and the Nucleolus~\cite{schmeidler:1163,aumann1985game}. Much work has also been done on the relation between solving cooperative games and linear programming and optimization~\cite{owen1975core,bilbao1999core,bilbao2001shapley,bilbao2008convexity,bilbao2002some,bilbao2000dual,bilbao2000shapley,bilbao1998axioms,bilbao1998banzhaf,bachrach2009cost,meir2011subsidies}. Others defined values for games, which as opposed to the core are a single imputation rather than a set of imputations. These include the Shapley value~\cite{shapley53value,shapley_shubik1954}, the Banzhaf index~\cite{banzhaf1965}, the Deegan-Packel index~\cite{deegan1978new} and many others~\cite{owen1977values,dubey1981value,owen1982modification,weber1988probabilistic,strafiin1988shapley,ruiz1998family,algaba2000position,monderer2002variations,algaba2010value}.

One type of a cooperative game, based on weights, is weighted majority games~\cite{isbell1956class,elgot1961truth,taylor1992characterization,taylor1999simple,matsui01npcompleteness,aumann2003endogenous}. In these games every player has a weight, similarly to CPGs. However, weighted majority games have a quota, and the value of a coalition is $1$ if the sum of the player weights exceeds the quota and is $0$ if it does not. Weighted majority games are a good model for voting in decision making bodies, so they are sometimes called weighted voting games, and many papers use game theoretic solution concepts to analyzing the power of players in these games~\cite{shapley_shubik1954,allingham1975economic,lucas1983measuring,widgren1994voting,bilbao2002voting,leech2002voting,algaba2004european,algaba2007distribution,aziz2009algorithmic}. Because of this important application, many researchers worked on computerized methods to calculate solutions and power in weighted majority games and their complexity~\cite{owen:mle_banzhaf:1975,arcaini1986algorithm,prasad1990np,deng_papa_1994,matsui01npcompleteness,albert2003voting,algaba2003computing,alonso2005generating,bachrach2010approximating,bilbao2000complexity,elkind2007computational,elkind2009computational}. Some research was also done on manipulations that change solutions in weighted majority games~\cite{yokoo2005coalitional,snyder2005legislative,bachrach2008divide,zuckerman-manipulating,aziz2009false,lasisi2010false,rey2010complexity}. 
CPGs are different from weighted majority games because they use the product operator and not the sum operator, and do not have a quota. \footnote{It is possible to define a quota for CPGs similar to weighted majority games, so the value of a coalition would be $1$ if the product of the weights exceed this value and $0$ otherwise, and this is a very interesting game to study in the future.}

Many other forms of ccoperative games were proposed in the past. Some game forms are based on sets and set operations~\cite{bilbao1998closure,elkind2008coalition,elkind2008dimensionality,van1992allocation,gilles1992games,bilbao2000bicooperative,bachrach2008coalitional,branzei2009coalitional,aziz2010monotone,bachrach2010coalitional}; 
Others forms rely on a logical representation and logical formulas in various languages~\cite{IeoSho05,conf/ijcai/YokooCSOI05,harrenstein2001boolean,aagotnes2006logic,journals/ai/WooldridgeD06,IeoSho06}; 
Some models are bases on a mathematical graph structures or an optimization problem in graphs~\cite{bilbao2006note,myerson1977graphs,kalai1982totally,deng_papa_1994,deng1998combinatorial,zolezzi2002transmission,nunez2004note,resnick2009cost,bachrach2008power,bachrach2009power,bachrach2010path,han2012game,igarashi2012average}; 
There are also many other forms based on other combinatorial structures~\cite{jimenez1998valores,bilbao:2000:cooperative,algaba2001myerson,fernandez2002generating,bilbao1998values,bilbao2003cooperative,algaba2003axiomatizations}.  
Many researchers also looked at the implications of cooperative game models to many fields, such as energy, auctions voting and networks~\cite{kamboj2011deploying,bilbao2000selectope,bachrach2011coalitional,baeyens2011wind,feldman2012stability}. 

\section{Conclusion}
\label{l_conc}

I showed an interesting type of cooperative game called CPG. In this game the utility of a coalition is the product of the weights of the players in the coalition. The key question I look at is how the players would agree to share the utility they make together. To do this I used well-known solution concepts from cooperative game theory. 

I gave some results about convexity and the core of these games: all these games never have dummy players, they are always convex and it is easy to compute imputations in the core. 

However, despite these results, many questions are still open about this kind of game. What are the connections to weighted majority games? How is it possible to calculate other solution concepts from game theory in this game, such as the Shapley value or other indexes? Can this game be extended by restricting the coalitions in some way or by adding a quota for the game? 


\bibliographystyle{abbrv}
\bibliography{cpg}

\end{document}